\documentclass[10pt, journal]{IEEEtran}
\usepackage{amssymb,amsmath}
\usepackage{graphics}
\usepackage{epsfig}
\usepackage{graphicx}
\usepackage{color}
\usepackage{ulem}
\usepackage{multirow}
\usepackage{hhline}
\usepackage{array}

\renewcommand{\baselinestretch}{0.9}
\usepackage[top=0.8in, bottom=0.7in, left=0.8in, right=0.8in]{geometry}
\newtheorem{theorem}{Theorem}
\newtheorem{lemma}{Lemma}

\begin{document}

% paper title
\title{Throughput and Delay Scaling in Supportive Two-Tier Networks}
\author{Long~Gao,\thanks{This work was supported in part by NSF under Grant 0726740, DoD-AFOSR under Grant FA9550-09-1-0107, DoD-DTRA under Grant HDTRA1-07-1-0037, the National University of Singapore under Grant 263-000-589-133, the China 863 Program under Grant 2011AA100706, the NSFC under Grant 60972073, the China National Great Science Specific Project under Grant 2010ZX03001-003, and the China 973 Program under Grant 2009CB320407.}\thanks{L. Gao and S. Cui are with the Department of Electrical and Computer
Engineering, Texas A\&M University, College Station, TX, 77843, USA (emails: lgao@tamu.edu, and cui@tamu.edu).}  \thanks{R. Zhang is with the Department of Electrical and Computer Engineering, National University of Singapore (e-mail: elezhang@nus.edu.sg). He is also with the Institute for Infocomm Research, A*STAR, Singapore (email:rzhang@i2r.a-star.edu.sg).}\thanks{C. Yin is with the Key Laboratory of Universal Wireless Communications, Ministry of Education, Beijing University of Posts and Telecommunications, China (email: ccyin@ieee.org).}Rui~Zhang,~Changchuan Yin, and
        Shuguang~Cui}

\maketitle

\begin{abstract}

Consider a wireless network that has two tiers with different priorities: a primary tier vs. a secondary tier, which is an emerging network scenario with the advancement of cognitive radio technologies. The primary tier consists of randomly distributed legacy nodes of density $n$, which have an absolute priority to access the spectrum. The secondary tier consists of randomly distributed cognitive nodes of density $m=n^\beta$ with $\beta\geq 2$, which can only access the spectrum opportunistically to limit the interference to the primary tier. Based on the assumption that the secondary tier is allowed to route the packets for the primary tier, we investigate the throughput and delay scaling laws of the two tiers in the following two scenarios: i) the primary and secondary nodes are all static; ii) the primary nodes are static while the secondary nodes are mobile. With the proposed protocols for the two tiers, we show that the primary tier can achieve a per-node throughput scaling of $\lambda_p(n)=\Theta\left(1/\log n\right)$ in the above two scenarios. In the associated delay analysis for the first scenario, we show that the primary tier can achieve a delay scaling of $D_p(n)=\Theta\left(\sqrt{n^\beta\log n}\lambda_p(n)\right)$ with $\lambda_p(n)=O\left(1/\log n\right)$. In the second scenario, with two mobility models considered for the secondary nodes: an i.i.d. mobility model and a random walk model, we show that the primary tier can achieve delay scaling laws of $\Theta(1)$ and $\Theta(1/S)$, respectively, where $S$ is the random walk step size. The throughput and delay scaling laws for the secondary tier are also established, which are the same as those for a stand-alone network.

\end{abstract}

\IEEEpeerreviewmaketitle

\section{Introduction}

The explosive growth of large-scale wireless applications
motivates people to study the fundamental limits over wireless
networks. Consider a randomly distributed wireless network with
density $n$ over a unit area, where the nodes are randomly grouped
into one-to-one source-destination (S-D) pairs. Initiated by the
seminal work in~\cite{Gupta:Capacity}, the throughput scaling laws
for such a network have been studied extensively in the
literature~\cite{Francheschetti:Closing}-\cite{Matthias:Mobility}.
For static networks, it is shown in~\cite{Gupta:Capacity} that the
traditional multi-hop transmission strategy can achieve a
throughput scaling of $\Theta\left(1/\sqrt{n\log
n}\right)$\footnote{We use the following notations throughout this
paper: i) $f(n)=O(g(n))$ means that there exists a constant $c$
and integer $N$ such that $f(n)<cg(n)$ for $n>N$; ii)
$f(n)=\Omega(g(n))$ means that $g(n)=O(f(n))$; iii)
$f(n)=\Theta(g(n))$ means that $f(n)=O(g(n))$ and $g(n)=O(f(n))$;
iv) $f(n)=o(g(n))$ means that $f(n)/g(n)\rightarrow 0$ as
$n\rightarrow\infty$; v) $f(n)=w(g(n))$ means that $g(n)=o(f(n))$.} per S-D pair. Such a
throughput scaling can be improved when the nodes are able to
move. It is shown in~\cite{Matthias:Mobility}-\cite{Gamal:Optimal}
that a per-node throughput scaling of $\Theta(1)$ is achievable in
mobile networks by exploring a special two-hop transmission scheme. Unfortunately, the throughput improvement in mobile networks
incurs a large packet
delay~\cite{Michael:Capacity},~\cite{Gamal:Optimal}, which is
another important performance metric in wireless networks. In
particular, it is shown in~\cite{Michael:Capacity} that the
constant per-node throughput is achieved at the cost of a delay
scaling of $\Theta(n)$. The delay-throughput tradeoffs for static and mobile
networks have also been investigated
in~\cite{Gamal:Optimal}-\cite{Ozgur:Throughput}. Specifically, it is shown in~\cite{Gamal:Optimal} that for the static network, the optimal delay-throughput tradeoff is given by $D(n)=\Theta(n\lambda(n))$ for $\lambda(n)=O\left(1/\sqrt{n\log n}\right)$, where $D(n)$ and $\lambda(n)$ are the delay and throughput per S-D pair, respectively; for the mobile network, in which nodes move according to a random walk (RW) model with a fixed step size $S=1/n$, a throughput of $\Theta(1)$ is achievable with the delay scaling as $\Theta(n\log n)$. In~\cite{Ying: Optimal}, the optimal delay-throughput tradeoffs under the RW node mobility model with an arbitrary step size $S$ is studied, where it is shown that the maximum throughput is $\Theta\left(\sqrt{D/n}\right)$ with $S=o(1)$ and $D=w(|\log S|/S^2)$.

The aforementioned literature mainly focuses on the delay and throughput
scaling laws for a single network. Recently, the emergence of
cognitive radio networks motives people to extend the result from
a single network to overlaid networks. Consider a
licensed primary network and a cognitive secondary network
coexisting over a unit area. The primary network has the absolute
priority to use the spectrum, while the secondary network can only
access the spectrum opportunistically to limit the interference to
the primary network. Based on such assumptions, a two-tier non-supportive network consisting of a primary tier and a secondary tier is considered in~\cite{Jeon:Cognitive}, where inter-tier packet relaying is not allowed. With an elegant transmission protocol, it is shown that by defining a preservation region around each primary node and assuming that the secondary tier knows the locations of all the primary nodes, both tiers can achieve the same throughput scaling law as a stand-alone wireless network in~\cite{Gupta:Capacity}, while the secondary tier may suffer from a finite outage probability. In~\cite{Yin:Scaling}, the same two-tier network setup as in~\cite{Jeon:Cognitive} is studied except that the secondary tier is assumed to only know the locations of the primary transmitters. It is shown that both tiers can still achieve the same delay-throughput tradeoffs as stand-alone networks in \cite{Gamal:Optimal}. Besides, the outage issue of the secondary tier is solved by introducing a new definition of the preservation region. However, such results are obtained
without considering possible positive interactions between the
primary network and the secondary network. In practice, the
secondary network, which is usually deployed after the existence
of the primary network for opportunistic spectrum access, can
transport data packets not only for itself but also for the
primary network due to their cognitive nature. As such, it is meaningful to investigate whether the
throughput and/or delay performance of the primary network (whose protocol was fixed before the deployment of the secondary tier) can be
improved with the opportunistic aid of the secondary network, while assuming the
secondary network still capable of keeping the same throughput and
delay scaling laws as the case where no supportive actions are
taken between the two networks.

In this paper, we define a \textit{supportive} two-tier network with a primary tier and a secondary tier as follows: The secondary tier is allowed to supportively relay the data
packets for the primary tier in an opportunistic way (i.e., the secondary users only utilize empty spectrum holes\footnote{As shown later, we introduce a concept of preservation region in the secondary protocol to ensure that only secondary nodes outside the preservation regions are allowed to transmit. As such, the spectrum holes refer to the spectrum resource outside the preservation regions.} in between primary transmissions even when they help with relaying the primary packets), whereas the primary tier is only required to transport its own data. Note that the potential security issues between the two tiers are important but not considered in this paper. Here we assume that the secondary nodes have the knowledge of the primary nodes' codebook, which is a common and reasonable assumption extensively used in the literature~\cite{Goldsmith:Breaking}-\cite{Devroye:Achievable}. Let $n$ and $m=n^{\beta}$ denote the node densities of the primary tier and the secondary tier, respectively. We investigate the throughput and delay scaling laws for
such a supportive two-tier network with $\beta\geq 2$ in the following two scenarios: i) the primary and secondary nodes are all static; ii) the primary nodes are static while the secondary nodes are mobile. With specialized protocols for the secondary tier, we show that the primary tier can achieve a per-node throughput scaling of $\lambda_p(n)=\Theta\left(1/\log n\right)$ in the above two scenarios with a classic time-slotted multi-hop transmission protocol similar to the one in~\cite{Gupta:Capacity}. In the associated delay analysis for the first scenario, we show that the primary tier can achieve a delay scaling of $D_p(n)=\Theta\left(\sqrt{n^\beta\log n}\lambda_p(n)\right)$ with $\lambda_p(n)=O\left(1/\log n\right)$. In the second scenario, with two mobility models considered for the secondary nodes: an i.i.d. mobility model and a random walk model, we show that the primary tier can achieve delay scaling laws of $\Theta(1)$ and $\Theta(1/S)$, respectively, where $S$ is the random walk step size. The throughput and delay scaling laws for the secondary tier are also established, which are the same as those for a stand-alone network. Note that generally speaking, $\beta$ could take any non-negative values. In this paper, we only consider the regime of $\beta\geq 2$ for analytical simplicity. As we will see later, such a condition is critical in the proofs of the delay and throughput results for the primary tier (i.e., \textit{Theorems 1-6}). We also want to point out that the results for the secondary tier are more general, which can hold in the regime of $\beta >1$.

Note that in~\cite{Gupta:Capacity}~\cite{Liu:Pernode}, the authors also pointed out that adding a large amount of extra pure relay nodes (which only relay traffic for other nodes), the throughput scaling can be improved at the cost of excessive network deployment. However, there are two key differences between their results and our results. First, in this paper, the added extra relays (the secondary nodes) only access spectrum opportunistically (i.e., they need not to be pre-allocated with any primary spectrum resource, given their cognitive nature), while the extra relay nodes mentioned in~\cite{Gupta:Capacity}~\cite{Liu:Pernode} are like regular primary nodes (just without generating their own traffic) who need to be assigned with certain primary spectrum resource. As such, based on the cognitive feature of the secondary nodes considered in this paper, the primary throughput improvement could be achieved in an existing primary network without the need to change the primary resource allocation scheme, while in~\cite{Gupta:Capacity}~\cite{Liu:Pernode}, the extra relay deployment has to be considered in the initial primary network design phase for its protocol to utilize the relays. In other words, the problem considered in this paper is how to improve the throughput scaling over an existing primary network by adding another supportive network tier (the secondary cognitive tier), where the primary network is already running a certain resource allocation scheme as we will discuss later in the paper, which is different from the networking scenarios considered in~\cite{Gupta:Capacity}~\cite{Liu:Pernode}. Second, in this paper, the extra relays are also source nodes on their own (i.e., they also initiate and support their own traffic within the secondary tier); and as one of the main results, we will show that even with their help to improve the primary-tier throughput, these extra relays (i.e., the secondary tier) could also achieve the same throughput scaling for their own traffic as a stand-alone network considered in~\cite{Gupta:Capacity}.

The rest of the paper is organized as follows. The system model is described and the main results are summarized in Section II. The proposed protocols for the primary
and secondary tiers are described in Section III. The delay and
throughput scaling laws for the primary tier are derived in
Section IV. The delay and throughput scaling laws for the
secondary tier are studied in Section V. Finally, Section VI
summarizes our conclusions.

\section{System Model and Main Results}

Consider a two-tier network with a primary tier and a denser secondary tier over a unit square. We assume that the nodes of the primary tier, so-called primary nodes, are static, and consider the following two scenarios: i) the nodes of the secondary tier, so-called secondary nodes, are also static; ii) the secondary nodes are mobile. We
first describe the network model, the interaction model between
the two tiers, the mobility models for the mobile secondary nodes in the second scenario, and the definitions of throughput and delay.
Then we summarize the main results in terms of the delay and throughput scaling laws for the proposed two-tier network.

\subsection{Network Model}

The primary nodes are distributed according
to a Poisson point process (PPP) of density $n$ and randomly
grouped into one-to-one source-destination (S-D) pairs. Likewise,
the secondary nodes are
distributed according to a PPP of density $m$ and randomly grouped
into S-D pairs. We assume that the density of the secondary tier
is higher than that of the primary tier, i.e.,
\begin{equation} \label{density}
m=n^\beta
\end{equation}
where we consider the case with $\beta\geq 2$. The primary tier and the
secondary tier share the same time, frequency, and space, but with
different priorities to access the spectrum: The former one is the
licensed user of the spectrum and thus has a higher priority; and
the latter one can only opportunistically access the spectrum to
limit the resulting interference to the primary tier, even when it helps with relaying the primary packets.

For the wireless channel, we only consider the large-scale
pathloss and ignore the effects of shadowing and small-scale
multipath fading. As such, the channel power gain $g(r)$ is given
as
\begin{equation} \label{pathloss}
g(r)=r^{-\alpha}
\end{equation}
where $r$ is the distance between the transmitter (TX) and the
corresponding receiver (RX), and $\alpha >2$ denotes the pathloss
exponent.

The ambient noise is assumed to be additive white Gaussian noise
(AWGN) with an average power $N_0$. During each time slot, we
assume that the $i$th primary TX-RX pair can achieve the following Shannon rate:
\begin{equation}  \label{prate}
R_p(i)=\log \left(1+\frac{P_p(i)g\left(\Vert
X_{p,tx}(i)-X_{p,rx}(i)\Vert \right)}{N_0+I_p(i)+I_{sp}(i)}\right)
\end{equation}
where the channel bandwidth is normalized to be unity for
simplicity, $\Vert \cdot \Vert$ denotes the norm operation,
$P_p(i)$ is the transmit power of the $i$th primary pair,
$X_{p,tx}(i)$ and $X_{p,rx}(i)$ are the TX and RX locations of
$i$th primary pair, respectively, $I_p(i)$ is the sum interference
from all other primary TXs, $I_{sp}(i)$ is the sum interference
from all the secondary TXs. Likewise, the data rate of the $j$th
secondary TX-RX pair is given by
\begin{equation} \label{srate}
R_s(j)=\log \left(1+\frac{P_s(j)g\left(\Vert
X_{s,tx}(j)-X_{s,rx}(j)\Vert \right)}{N_0+I_s(j)+I_{ps}(j)}\right)
\end{equation}
where $P_s(j)$ is the transmit power of the $j$th secondary pair,
$X_{s,tx}(j)$ and $X_{s,rx}(j)$ are the TX and RX locations of the
$j$th secondary pair, respectively, $I_s(j)$ is the sum
interference from all other secondary TXs to the RX of the $j$th
secondary pair, and $I_{ps}(j)$ is the sum interference from all
primary TXs.

\subsection{Interaction Model}

As shown in the previous
work~\cite{Jeon:Cognitive},~\cite{Yin:Scaling}, although the
opportunistic data transmission in the secondary tier does not
degrade the scaling law of the primary tier, it may reduce the
throughput in the primary tier by a constant factor due to the
fact that the interference from the secondary tier to the
primary tier cannot be reduced to zero. To completely
compensate the throughput degradation or even improve the
throughput scaling law of the primary tier in the two-tier setup,
we could allow certain positive interactions between the two
tiers. Specifically, we assume that the secondary nodes are
willing to act as relay nodes for the primary tier, while the
primary nodes are not assumed to do so. When a primary source node
transmits packets, the surrounding secondary nodes could pretend
to be primary nodes to relay the packets (which is feasible since they are software-programmable cognitive radios). Note that, these ``fake''
primary nodes do not have the same priority as the real primary
nodes in terms of spectrum access, i.e., they can only use the
spectrum opportunistically in the same way as a regular secondary
node. The assumption that the secondary tier is allowed to relay the primary packets is the essential difference between our model and the models
in~\cite{Jeon:Cognitive},~\cite{Yin:Scaling}.

\subsection{Mobility Model}

In the scenario where the secondary nodes are mobile, we assume that the positions of the primary nodes are fixed
whereas the secondary nodes stay static in one primary time
slot\footnote{As we will see in Section III, the data transmission
is time-slotted in the primary and secondary tiers.} and change
their positions at the next slot. In particular, we consider the
following two mobility models for the secondary nodes.

\textbf{Two-dimensional i.i.d. mobility model}~\cite{Michael:Capacity}: The secondary nodes are uniformly and randomly distributed in the unit area at each primary time slot. The node locations are independent of each other, and independent from time slot to time slot, i.e., the nodes are totally reshuffled over each primary time slot.

\textbf{Two-dimensional RW model}~\cite{Gamal:Optimal},~\cite{Ying: Optimal}: We divide the unit square into $1/S$ small-square RW-cells, each of them with size $S$. The RW-cells are indexed by $(x,y)$, where $x,y\in\{1,2,\cdots,1/\sqrt{S}\}$. A secondary node that stays in a RW-cell at a particular primary time slot will move to one of its eight neighboring RW-cells at the next slot with equal probability (i.e., 1/8). For the convenience of analysis, when a secondary node hits the boundary of the unit square, we assume that it jumps over the opposite edge to eliminate the edge effect~\cite{Gamal:Optimal},~\cite{Ying: Optimal}. The nodes within a RW-cell are uniformly and randomly distributed. Note that the unit square are also divided into primary cells and secondary cells in the proposed protocols as discussed in Section III, which are different from the RW-cells defined
  above. In this paper, we only consider the case where the size of the
  RW-cell is greater than or equal to that of the primary cell.

\subsection{Throughput and Delay}

The \textit{throughput per S-D pair} (per-node throughput) is
defined as the average data rate that each source node can
transmit to its chosen destination as
in~\cite{Jeon:Cognitive},~\cite{Yin:Scaling}, which is asymptotically determined by the network density. Besides, the \textit{sum throughput} is defined
as the product between the throughput per S-D pair and the number
of S-D pairs in the network. In the following, we use
$\lambda_{p}(n)$ and $\lambda_{s}(m)$ to denote the throughputs
per S-D pair for the primary tier and the secondary tier,
respectively; and we use $T_{p}(n)$ and $T_{s}(m)$ to denote the
sum throughputs for the primary tier and the secondary tier,
respectively.

The delay of a primary packet is defined as the average number of
primary time slots that it takes to reach the primary destination
node after the departure from the primary source node. Similarly,
we define the delay of a secondary packet as the average number of
secondary time slots for the packet to travel from the secondary
source node to the secondary destination node. We use $D_p(n)$ and
$D_s(m)$ to denote packet delays for the primary tier and the
secondary tier, respectively. For simplicity, we use a fluid
model~\cite{Gamal:Optimal} for the delay analysis, in which we
divide each time slot to multiple packet slots and the size of the
data packets can be scaled down with the increase of network
density.

\subsection{Main Results}

We summarize the main results in terms of the throughput and delay scaling laws for the supportive two-tier network here. We first present the results for the scenario where the primary and secondary nodes are all static and then describe the results for the scenario with mobile secondary nodes.
\begin{description}
  \item [i)]\hspace{-8mm}The primary and secondary nodes are all static.
  \begin{itemize}
    \item It is shown that the primary tier can achieve a per-node throughput scaling of $\lambda_p(n)=\Theta\left(1/\log n\right)$ and a delay scaling of $D_p(n)=\Theta\left(\sqrt{n^\beta\log n}\lambda_p(n)\right)$ for $\lambda_p(n)=O\left(1/\log n\right)$.
    \item It is shown that the secondary tier can achieve a per-node throughput scaling of $\lambda_s(m)=\Theta\left(\frac{1}{\sqrt{m\log m}}\right)$ and a delay scaling of $D_s(m)=\Theta(m\lambda_s(m)),~\textrm{for}~\lambda_s(m)=O\left(\frac{1}{\sqrt{m\log m}}\right)$.
  \end{itemize}
  \item [ii)]\hspace{-7mm}The primary nodes are static and the secondary nodes are mobile.
  \begin{itemize}
    \item It is shown that the primary tier can achieve a per-node throughput scaling of $\lambda_p(n)=\Theta\left(1/\log n\right)$, and delay scaling laws of $\Theta(1)$ and $\Theta(1/S)$ with the i.i.d. mobility model and the RW mobility model, respectively.
    \item It is shown that the secondary tier can achieve a per-node throughput scaling of $\lambda_s(m)=\Theta(1)$, and delay scaling laws of $\Theta(m)$ and $\Theta\left(m^2S\log\frac{1}{S}\right)$ with the i.i.d. mobility model and the RW mobility model, respectively.
  \end{itemize}
\end{description}

\section{Network Protocols}
In this section, we describe the proposed protocols for the
primary tier and the secondary tier, respectively. The primary
tier deploys a modified time-slotted multi-hop transmission scheme
from those for the primary network
in~\cite{Jeon:Cognitive},~\cite{Yin:Scaling}, while the secondary
tier chooses its protocol according to the given primary transmission scheme. In
the following, we use $p(E)$ to represent the probability of event
$E$, and claim that an event $E_n$ occurs with high probability
(w.h.p.) if $p(E_n)\to 1$ as $n\to \infty$.

\subsection{The Primary Protocol}
The main sketch of the primary protocol is given as follows:

\noindent i) Divide the unit square into small-square primary cells with
size $a_p(n)$. In order to maintain the full connectivity within
the primary tier even without the aid of the secondary tier and enable the possible support from the secondary tier (see \textit{Theorem 1} for details), we
have $a_p(n)\geq \sqrt 2\beta\log n/n$ such that each cell has at least one
primary node w.h.p..

\noindent ii) Group every $N_c$ primary cells into a primary cluster. The
cells in each primary cluster take turns to be active in a
round-robin fashion. We divide the transmission time into TDMA
frames, where each frame has $N_c$ primary time slots that correspond
to the number of cells in each primary cluster. Note that the
number of primary cells in a primary cluster has to satisfy $N_c\geq 64$ such that we can appropriately arrange the preservation
regions and the collection regions, which will be formally defined later
in the secondary protocol. For convenience, we take $N_c=64$ throughout the paper.

\noindent iii) Define the S-D data path along which the packets are routed from
the source node to the destination node: The data path follows a
horizontal line and a vertical line connecting the source node and
the destination node, which is the same as that defined
in~\cite{Jeon:Cognitive},~\cite{Yin:Scaling}. Pick an arbitrary
node within a primary cell as the designated relay node, which is
responsible for relaying the packets of all the data paths passing
through the cell.

\noindent iv) When a primary cell is active, each primary source node in it
takes turns to transmit one of its own packets with probability
$p$. The parameter $p$ is used for access control of the primary packets such that the queues in the mobile secondary nodes are stable. As shown later in \textit{Theorems 5} and \textit{6}, $p$ can take any positive values less than one. Afterwards, the designated relay node transmits one packet
for each of the S-D paths passing through the cell. Note that a primary source node could also be a designated relay node. If this is the case, the source node first sends one packet of its own and then sends the packets for the other primary nodes. The above
packet transmissions follow a time-slotted pattern within the active
primary time slot, which is divided into packet slots as shown in Fig.~\ref{primary protocol}. Each source
node reserves a packet slot no matter it transmits or not. If the
designated relay node keeps silent, i.e., has no packets to transmit, it does not
reserve any packet slots. For each packet, if the destination
node is found in the adjacent cell, the packet will be directly
delivered to the destination. Otherwise, the primary transmitter blindly broadcasts the packet to its neighboring nodes and it is the responsibility of the designated relay node in one adjacent cell along the data path to store the packet for future transmissions. At each packet transmission, the TX node transmits with
power of $Pa_{p}^{\frac{\alpha}{2}}(n)$, where $P$ is a constant.

\noindent v) We assume that all the packets for each S-D pair are labelled
with serial numbers (SNs). The following handshake mechanism is
used when a TX node is scheduled to transmit a packet to a
destination node: The TX sends a request message to initiate the
process; the destination node replies with the desired SN; if the
TX has the packet with the desired SN, it will send the packet to
the destination node; otherwise, it stays idle. As we will see in
the proposed secondary protocol for the scenario with mobile secondary nodes, the helping secondary relay nodes will
take advantage of the above handshake mechanism to remove the
outdated (already-delivered) primary packets from their queues. We
assume that the length of the handshake message is negligible
compared to that of the primary data packet in the throughput
analysis for the primary tier as discussed in Section IV.

\begin{figure}
\centering
\includegraphics[width=2.5in]{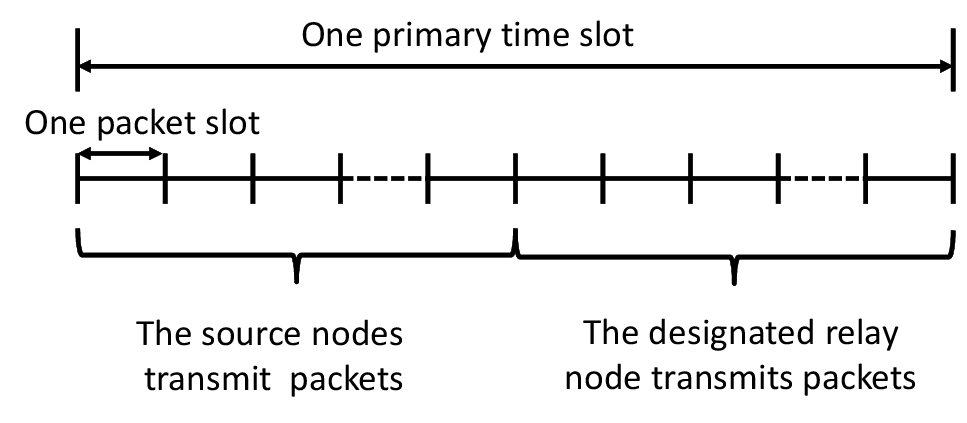}
\caption{The illustration of the primary protocol.}
\label{primary protocol}
\end{figure}

Note that running of the above protocol for the primary tier is
independent of whether the secondary tier is present or not. When
the secondary tier is absent, the primary tier can achieve the
throughput scaling law as a stand-alone network discussed in~\cite{Gupta:Capacity}. When the
secondary tier is present as shown in Section IV, the primary tier
can achieve a better throughput scaling law with the aid of the
secondary tier.

\subsection{The Secondary Protocol}

In this section, we present the proposed secondary protocol for the following two scenarios: i) the scenario with static secondary nodes, and ii) the scenario with mobile secondary nodes. In the scenario where the primary and secondary nodes are all static, the secondary nodes chop the received primary packets into smaller pieces suitable for secondary-tier transmissions. The small data pieces will be reassembled before they are delivered to the primary destination nodes. In the scenario where the secondary nodes are mobile, the received packets are stored in the secondary nodes and delivered to the corresponding primary destination node only when the secondary nodes move into the neighboring area of the primary destination node. In both scenarios, such helps are achieved with the secondary nodes opportunistically exploring the primary spectrum without hurting the original primary performance. As such, the primary tier is expected to achieve better throughput and/or delay scaling laws.

\vspace{0.25cm}
\noindent\textbf{Protocol for Static Secondary Tier}
\begin{figure}
\centering
\includegraphics[width=2.5in]{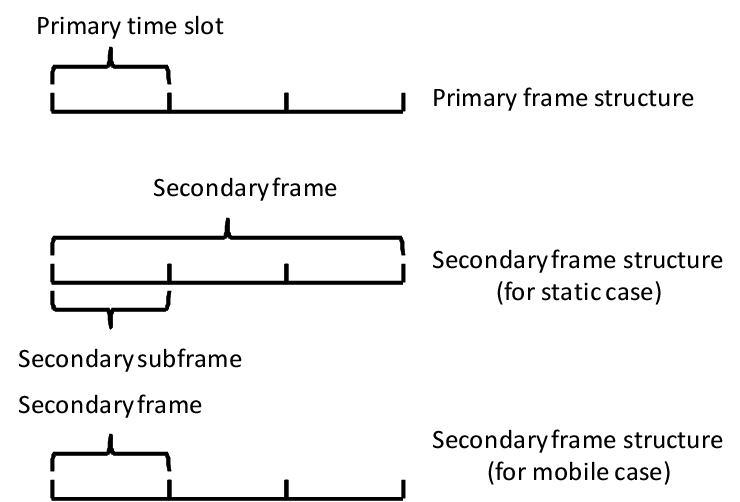}
\caption{Frame relationship between the two tiers.}
\label{frame}
\end{figure}

We assume that the secondary nodes have the necessary cognitive
features such as software-programmability to ``pretend'' as primary nodes such that they could be
chosen as the designated primary relay nodes within a particular
primary cell. As later shown by \textit{Lemma~\ref{lemma2}} in Section IV,
a randomly selected designated relay node for the primary packet
in each primary cell is a secondary node w.h.p.. Once a secondary node is chosen and fixed to be a designated relay node, it keeps listening instead of relaying primary packets when its associated primary cell is active\footnote{Actually, none of the secondary nodes within a active primary cell are allowed to transmit according to the secondary protocol as will be shown later., such that only primary source nodes within the cell transmit packets. As will be shown later in the proof of \textit{Theorem 1}, this operation will significantly improve the throughput of the primary tier. In this scenario, the primary packets are relayed not only by the designated secondary relay nodes but also by other secondary nodes, which will be explained in details next.}

We use the time-sharing technique to guarantee successful packet
deliveries from the secondary nodes to the primary destination
nodes as follows. We divide each secondary frame into three
equal-length subframes, such that each of them has the same length
as one primary time slot as shown in~Fig.~\ref{frame}. The first subframe is used to transmit
the secondary packets within the secondary tier. The second
subframe is used to relay the primary packets to the next relay
nodes. Accordingly, the third subframe of each secondary frame is
used to deliver the primary packets from the intermediate
destination nodes\footnote{An ``intermediate'' destination node of
a primary packet within the secondary tier is a chosen secondary
node in the primary cell within which the final primary
destination node is located.} in the secondary tier to their final
destination nodes in the primary tier. Specifically, for the first
subframe, we use the following protocol:
\begin{itemize}
\item Divide the unit area into square secondary cells with size
$a_s(m)$. In order to maintain the full connectivity within the
secondary tier, we have to guarantee $a_s(m) \geq 2\log m/m$ with
a similar argument to that in the primary tier. Given the assumption of $\beta\geq 2$, the size of the secondary cell is much smaller than that of the primary cell, i.e., $a_s(m)\ll a_p(n)$.

\item Group the secondary cells into secondary clusters, with each
secondary cluster of 64 cells. Each secondary cluster also follows
a 64-TDMA pattern to communicate, which means that the first subframe is
divided into 64 secondary time slots.

\item Define a preservation region as nine primary cells centered
at an active primary TX and a layer of secondary cells around
them, shown as the square with dashed edges in
Fig.~\ref{preservation}. Only the secondary TXs in an active
secondary cell outside all the preservation regions can transmit
data packets; otherwise, they buffer the packets until the
particular preservation region is cleared. When an active
secondary cell is outside the preservation regions in the first
subframe, it allows the transmission of one packet for each
secondary source node and for each S-D path passing through the cell
in a time-slotted pattern within the active secondary time slot w.h.p.. The routing of secondary packets follows similarly defined data paths as those in the primary tier.

\item At each transmission, the active secondary TX node can only
transmit to a node in its adjacent cells with power of $P
a_{s}^{\frac{\alpha}{2}}(m)$.
\end{itemize}

In the second subframe, only secondary nodes who carry primary
packets take the time resource to transmit. Note that each primary
packet is broadcasted from the primary source node to its
neighboring primary cells where we assume that there are $N$
secondary nodes including the designated secondary relay node in the neighboring cell along the primary data path that successfully decodes the packet and ready to relay. Since the density of the secondary nodes is larger than that of the primary nodes, the throughput per secondary S-D pair is less than that per primary S-D pair as shown later in \textit{Theorem 7}. As such, packet splitting is needed to ensure that there is no bottleneck effect in relaying primary packets through the secondary tier. In particular,
each secondary node relays $1/N$ portion of the primary packet to
the intermediate destination node in a multi-hop fashion\footnote{We assume that there exists a central entity to coordinate the transmissions of the $N$ packet segments such that each chosen secondary node relays a unique portion of the primary packet.}, and the
value of $N$ is set as
\begin{equation}\label{relaynumber}
    N=\Theta\left(\sqrt{\frac{m}{\log m}}\right).
\end{equation}
From \textit{Lemma 1} in Section IV, we can guarantee that there are
more than $N$ secondary nodes in each primary cell w.h.p. when $\beta\geq 2$. The specific transmission
scheme in the second subframe is the same as that in the first
subframe, where the subframe is divided into 64 time
slots and all the traffic is for primary packets.

\begin{figure}
\centering
\includegraphics[width=2.5in]{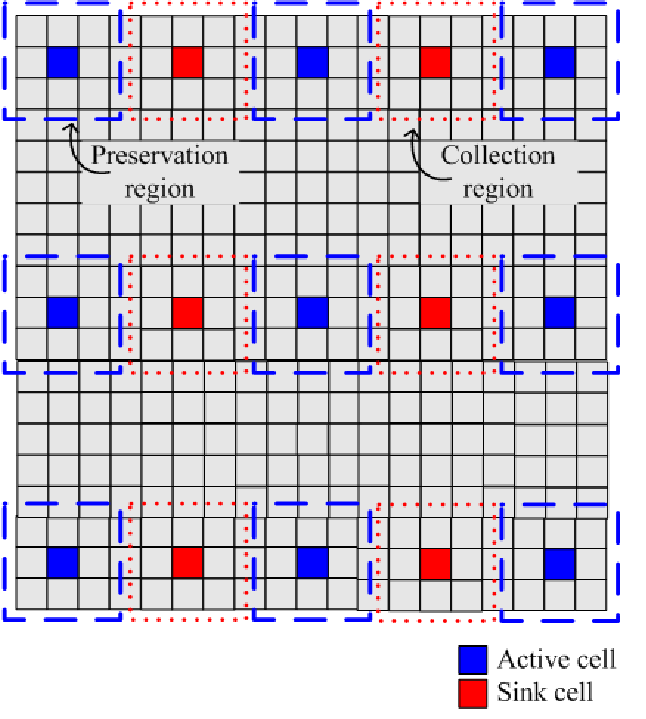}
\caption{Preservation regions and collection regions.}
\label{preservation}
\end{figure}

At the intermediate destination nodes, the received primary packet
segments are reassembled into the original primary packets. Then in
the third subframe, we use the following protocol to deliver the
packets to the primary destination nodes:
\begin{itemize}
\item Define a collection region as nine primary cells and a layer
of secondary cells around them, shown as the square with dotted
edges in Fig.~\ref{preservation}, where the collection region is
located between two preservation regions along the horizontal line
and they are not overlapped with each other.

\item Deliver the primary packets from the intermediate
destination nodes in the secondary tier to the corresponding
primary destination nodes in the sink cell, which is defined as
the center primary cell of the collection region. The primary
destination nodes in the sink cell take turns to receive data by
following a time-slotted pattern, where the corresponding intermediate
destination node in the collection region transmits by pretending
as a primary TX node. Given that the third subframe is of an equal
length to one primary slot, each primary destination node in the
sink cell can receive one primary packet from the corresponding
intermediate destination node.

\item At each transmission, the intermediate destination node
transmits with the same power as that for a primary node, i.e., $P
a_{p}^{\frac{\alpha}{2}}(n)$.
\end{itemize}
\vspace{0.25cm}
\noindent\textbf{Protocol for Mobile Secondary Tier}

Like in the scenario with static secondary nodes, we assume that the secondary nodes have the necessary cognitive
features to ``pretend'' as primary nodes such that they could be
chosen as the designated primary relay nodes within a particular
primary cell. Similar to the protocol for the static secondary tier, once a secondary node is chosen to be a designated relay node, it is required to keep listening instead of relaying primary packets when it jumps into an active primary cell. In this scenario, the primary packets are jointly relayed by the designated relay nodes and other secondary nodes in a special way, which is described next.

Divide the transmission time into TDMA frames, where the
secondary frame has the same length as that of one primary time slot as shown in~Fig.~\ref{frame}. To limit the interference to primary transmissions, we define
preservation regions in a similar way to that in the scenario with static secondary nodes. To faciliate the description of the secondary protocol, we define the \textit{separation threshold time} of random walk as~\cite{Aldous:Reversible}
\begin{equation}\label{separation time}
    \tau=\min\{t: s(t)\leq e^{-1}\}
\end{equation}
where $s(t)$ measures the separation from the stationary
distribution at time $t$, which is given by
\begin{eqnarray}
% \nonumber to remove numbering (before each equation)
  \nonumber s(t) &=& \min\bigg\{s:p_{(x,y),(u,v)}(t)\geq
  (1-s)\pi_{(u,v)},
  \\
  &&~\textrm{for all}~x,y,u,v\in\{1,2,\cdots,1/\sqrt{S}\}\bigg\}
\end{eqnarray}
where $p_{(x,y),(u,v)}(t)$ denotes the probability that a
secondary node hits RW-cell $(u,v)$ at time $t$ starting from
RW-cell $(x,y)$ at time $0$, and $\pi_{(u,v)}=S$ is the probability of staying at RW-cell $(u,v)$ at the stationary state. We have
$\tau=\Theta(1/S)$~\cite{Aldous:Reversible}.

The secondary nodes perform one of the the following two operations
according to whether they are in the preservation regions or not:

\noindent i) If a secondary node is in a preservation region, it is not allowed to transmit packets. Instead, it receives the packets from the active primary transmitters and store them in the buffer for future deliveries. Each secondary node maintains $Q$ separate queues for each primary S-D
  pair. For the i.i.d. mobility model, we take $Q=1$, i.e., only one queue is needed for each primary
  S-D pair. For the RW model, $Q$ takes the value of $\tau$ given by~(\ref{separation time}). The packet received at time slot $t$ is
  considered to be `type $k$' and stored in the $k$th queue, if
  $\left\{\lfloor\frac{t}{64}\rfloor~\textrm{mod}~Q\right\}=k$, where $\lfloor x \rfloor$ denotes the flooring operation.

\noindent~ii) If a secondary node is not in a preservation region, it transmits the primary and secondary packets in the buffer. In order to guarantee successful deliveries for both primary and secondary packets, we evenly and randomly divide the secondary S-D pairs into two classes: Class I and Class II. Define a collection region in a similar way to that in the scenario with static secondary nodes. In the following, we
describe the operations of the secondary nodes of Class I based on
whether they are in the collection regions or not. The secondary
nodes of Class II perform a similar task over switched timing
relationships with the odd and even primary time slots.
\begin{itemize}
        \item If the secondary nodes are in the collection regions, they keep silent at the odd primary time slots and deliver the primary packets at the even primary time slots to the primary destination nodes in the sink cell, which is defined as the center primary cell of the collection region. In a particular primary time slot, the primary destination nodes in the sink cell take turns to receive packets following a time-slotted pattern. For a particular primary destination node at time $t$, we choose an arbitrary secondary node in the sink cell to send a request message to the destination node. The destination node replies with the desired SN, which will be heard by all secondary nodes within the nine primary cells of the collection region. These secondary nodes remove all outdated packets for the destination node, whose SNs are lower than the desired one. For the i.i.d. mobility model, if one of these secondary nodes has the packet with the desired SN and it is in the sink cell, it sends the packet to the destination node. For the RW model, if one of these secondary nodes has the desired packet in the $k$th queue with $k=\left\{\lfloor\frac{t}{64}\rfloor~\textrm{mod}~Q\right\}$ and it is in the sink cell, it sends the packet to the destination node. At each transmission, the secondary node transmits with the same power as that for a primary node, i.e., $Pa_{p}^{\frac{\alpha}{2}}(n)$.
        \item If the secondary nodes are not in the collection regions, they keep silent at the even primary time slots and transmit secondary packets at the odd primary time slots as follows. Divide the unit square into small-square secondary cells with size $a_s(m)=1/m$ and group every 64 secondary cells into a secondary cluster. The cells in each secondary cluster take turns to be active in a round-robin fashion. In a particular active secondary cell, we could use Scheme 2 in~\cite{Gamal:Optimal} to transmit secondary packets with power of $Pa_s^{\frac{\alpha}{2}}(m)$ within the secondary tier.
      \end{itemize}

\section{Throughput and Delay Analysis for the Primary Tier}

In the following, we first present the throughput and delay scaling laws for the primary tier in the scenario where the primary and secondary nodes are all static, and then discuss the scenario where the secondary nodes are mobile.

\subsection{The Scenario with Static Secondary Nodes}

We first give the throughput and delay scaling
laws for the primary tier, followed by the delay-throughput
tradeoff.

\vspace{0.25cm}
\noindent\textbf{Throughput Analysis}

In order to obtain the throughput scaling law, we first give the
following lemmas.
\begin{lemma}\label{lemma1}
The numbers of the primary nodes and secondary nodes in each
primary cell are $\Theta(na_p(n))$ and $\Theta(ma_p(n))$ w.h.p.,
respectively.
\end{lemma}

This is an existing result. The proof can be found in \cite{Yin:Scaling}.

\begin{lemma}\label{lemma2}
If the secondary nodes compete to be the designated relay nodes
for the primary tier by pretending as primary nodes, a randomly
selected designated relay node for the primary packet in each primary
cell is a secondary node w.h.p..
\end{lemma}
\begin{proof}
Let $\eta$ denote the probability that a randomly
selected designated relay node for the primary packet in a particular primary
cell is a secondary node. We have $\eta=\frac{\Theta(ma_p(n))}{\Theta(ma_p(n)+na_p(n))}$ from \textit{Lemma 1}, which approaches one as $n\rightarrow\infty$. This completes the proof.
\end{proof}

\begin{lemma}\label{lemma3}
With the protocols given in Section III, an active
primary cell can support a constant data rate of $K_1$, where
$K_1>0$ independent of $n$ and $m$.
\end{lemma}

The proof can be found in Appendix I.

\begin{lemma}\label{lemma4}
With the protocols given in Section III, the secondary
tier can deliver the primary packets to the intended primary
destination node at a constant data rate of $K_2$, where $K_2>0$
independent of $n$ and $m$.
\end{lemma}

The proof can be found in Appendix I.

Note that \textit{Lemmas 2-4} are new results for the supportive two-tier network setup. Based on \textit{Lemmas 1-4}, we have the following theorem.

\begin{theorem}\label{pthroughput}
With the protocols given in Section III, the primary tier can
achieve the following throughput per S-D pair and sum throughput
w.h.p. when $\beta\geq 2$:
\begin{equation}\label{pthroughput1}
    \lambda_p(n)=\Theta\left(\frac{1}{n a_p(n)}\right)
\end{equation}
and
\begin{equation}\label{pthroughput2}
    T_p(n)=\Theta\left(\frac{1}{a_p(n)}\right),
\end{equation}
where $a_p(n)\geq \sqrt 2\beta\log n/n$ and $a_p(n)=o(1)$.
\end{theorem}
\begin{proof}
In this proof, we first derive an upper-bound of the throughput per S-D pair and then provide a lower-bound by using the proposed protocol in Section III. As shown later, these two bounds will give us the exact result given in~(\ref{pthroughput1}).

We first derive the upper-bound. From \textit{Lemmas 3} and \textit{4}, we know that the primary TX can pour
its packets into the secondary tier at a constant rate $K=\min(K_1,K_2)$. Since all the designated relay nodes are secondary nodes w.h.p. and they keep silent when the primary cells that they are located in become active, only primary source nodes transmit packets in each active primary cell. Given that the number of the primary source nodes in each primary cell is
of $\Theta(na_p(n))$ as shown in \textit{Lemma 1}, the upper-bound of the throughput per S-D pair is of
$\Theta(K/na_p(n))=\Theta(1/na_p(n))$.

For the lower-bound, we calculate the achievable throughput per S-D pair with the proposed protocols. In the proposed protocols, each primary source node pours all its packets
into the secondary tier w.h.p. (from \textit{Lemma 2}) by splitting data into
$N=\Theta\left(\sqrt{m/\log m}\right)$ secondary data paths, each of
them at a rate of $\Theta(\frac{1}{m\sqrt{a_s(m)}})$ given in~(\ref{sthroughput1s}). Set
$\sqrt{a_s(m)}=\frac{na_p(n)}{\sqrt{m\log m}}$, which satisfies $a_s(m)\geq 2\log m/m$. As such, each primary
source node achieves a throughput scaling law of $N\Theta(\frac{1}{m\sqrt{a_s(m)}})=\Theta\left(1/na_p(n)\right)$, which can be considered as a lower-bound of the throughput per S-D pair.

By combining the two bounds, (\ref{pthroughput1}) is proved. Since the total number of primary nodes in the unit square
is of $\Theta(n)$ w.h.p., we have $T_p(n)=\Theta(n\lambda_p(n))=\Theta\left(1/a_p(n)\right)$ w.h.p.. This completes the proof.
\end{proof}

Note that the condition of $\beta\geq 2$ is needed in the proof to guarantee that there are more than $N$ secondary nodes in each primary cell w.h.p.. By setting $a_p(n)=\sqrt 2\beta\log n/n$, the primary tier can achieve the
following throughput per S-D pair and sum throughput w.h.p.:
\begin{equation}\label{compare1}
    \lambda_p(n)=\Theta\left(\frac{1}{\log n}\right)
\end{equation}
and
\begin{equation}
    T_p(n)=\Theta\left(\frac{n}{\log n}\right).
\end{equation}

From (\ref{compare1}), we see that the per-node throughput scaling law of the primary tier can be improved from $\Theta\left(1/(\sqrt {n\log n})\right)$ as in the stand-alone network to $\Theta\left(1/\log n)\right)$ with the help of the secondary tier.

\vspace{0.25cm}
\noindent\textbf{Delay Analysis}

We now analyze the delay performance of the primary tier with the aid
of a static secondary tier. In the proposed protocols, we know that the
primary tier pours all the primary packets into the secondary tier
w.h.p. based on \textit{Lemma 2}. In order to analyze the delay of
the primary tier, we have to calculate the traveling time for the $N$
segments of a primary packet to reach the corresponding
intermediate destination node within the secondary tier. Since the
data paths for the $N$ segments are along the route and an active
secondary cell (outside all the preservation regions) transmits
one packet for each data path passing through it within a secondary
time slot, we can guarantee that the $N$ segments depart from the
$N$ nodes, move hop by hop along the data paths, and finally
reach the corresponding intermediate destination node in a
synchronized fashion. According to the definition of packet
delay, the $N$ segments
experience the same delay later given in~(\ref{sdelay}) within the
secondary tier, and all the segments arrive the intermediate destination node within one secondary slot.

Let $L_p$ and $L_s$ denote the durations of the primary and
secondary time slots, respectively. According to the proposed
protocols, we have
\begin{equation}
    L_p=64L_s.
\end{equation}
Since we split the secondary time frame into three fractions and
use one of them for the primary packet relaying, each primary
packet suffers from the following delay:
\begin{equation}\label{pdelay}
    D_p(n)=\frac{3}{64}D_s(m)+C=\Theta\left(\frac{1}{\sqrt{a_s(m)}}\right)
\end{equation}
where the secondary-tier delay $D_s(m)$ is later derived in~(\ref{sdelay}), $C$ denotes the average time for a primary packet to travel
from the primary source node to the $N$ secondary relay nodes plus that
from the intermediate destination node to the final destination
node, which is a constant. We see from (\ref{pdelay}) that the
delay of the primary tier is only determined by the size of the
secondary cell $a_s(m)$. In order to obtain a better delay
performance, we should make $a_s(m)$ as large as possible.
However, a larger $a_s(m)$ results in a decreased throughput per
S-D pair in the secondary tier and hence a decreased throughput
for the primary tier, for the primary traffic traverses over
the secondary tier w.h.p.. In Appendix III, we derive the
relationship between $a_p(n)$ and $a_s(m)$ in our supportive
two-tier setup as
\begin{equation}\label{cellsize}
    a_s(m)=\frac{n^2a_p^2(n)}{m\log m}
\end{equation}
where we have $a_s(m)\geq 2\log m/m$ when $a_p(n)\geq \sqrt 2\beta\log n/n$.

Substituting~(\ref{cellsize}) into (\ref{pdelay}), we have the following theorem.
\begin{theorem}
According to the proposed protocols in Section III, the primary
tier can achieve the following delay w.h.p. when $\beta\geq 2$.
\begin{equation}\label{pdelayfinal}
D_p(n)=\Theta\left(\frac{\sqrt{m\log
m}}{na_p(n)}\right)=\Theta\left(\frac{\sqrt{n^\beta\log
n}}{na_p(n)}\right).
\end{equation}
\end{theorem}

\vspace{0.25cm}
\noindent\textbf{Delay-Throughput Tradeoff}

Combining the results in~(\ref{pthroughput1}) and
(\ref{pdelayfinal}), the delay-throughput tradeoff for the primary
tier is given by the following theorem.
\begin{theorem}\label{ptradeoff}
With the protocols given in Section III, the delay-throughput
tradeoff in the primary tier is given by
\begin{equation}\label{tradeoff}
   \hspace{-3 mm} D_p(n)=\Theta\left(\sqrt{n^\beta\log n}\lambda_p(n)\right)~\textrm{for}~\lambda_p(n)=O\left(\frac{1}{\log n}\right).
\end{equation}
\end{theorem}

In Fig.~\ref{tradeoff1}, we draw the delay-throughput tradeoff for the primary tier with the aid of static secondary nodes compared with the optimal result without the secondary tier as shown in~\cite{Gamal:Optimal}. In the figure, the line segments PR and PQ denote the delay-throughput tradeoffs for the primary tier and a stand-alone network, respectively, where the scales of the axes are in terms of the orders in $n$. Any delay-throughput pair in the two segments can be achieved by adjusting the size of the primary cell. We see that with the aid of the secondary nodes, the primary tier can achieve higher order of throughput scaling. However, in the regime of $\lambda (n)=O(1/\sqrt{n\log n})$, the delay scaling of the primary tier with the aid of the secondary tier is worse than that without the secondary tier, i.e., the derived delay-throughput tradeoff given in (\ref{tradeoff}) is strictly suboptimal in this scenario. In other words, the primary tier can has better throughput performance with sacrificing certain delay performance in this scenario.

\begin{figure}
\centering
\includegraphics[width=2.5in]{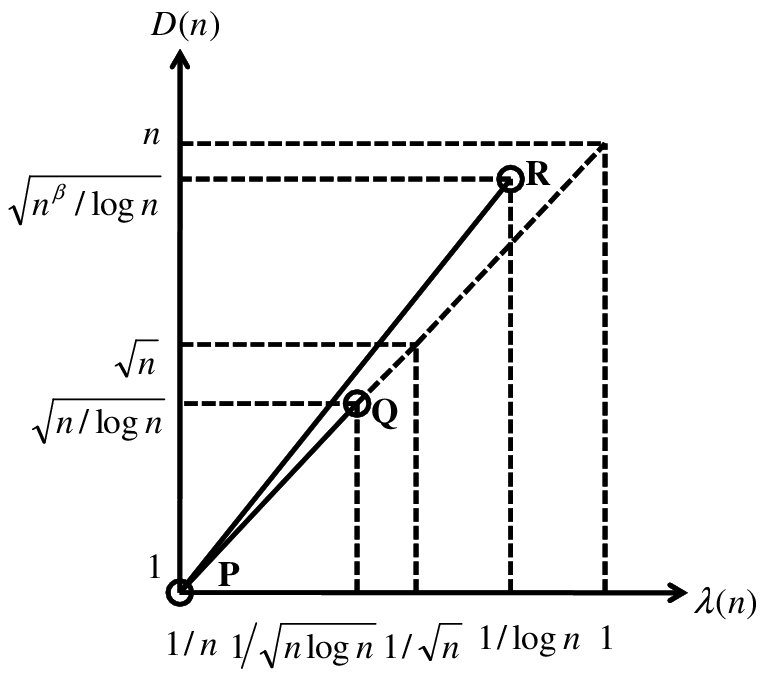}
\caption{Delay-throughput tradeoff for the primary tier with the aid of static secondary nodes.}
\label{tradeoff1}
\end{figure}

\subsection{The Scenario with Mobile Secondary Nodes}

\vspace{0.25cm}
\noindent\textbf{Throughput Analysis}

In order to obtain the throughput scaling law, we first give the
following lemmas.

\begin{lemma}\label{lemma5}
With the protocols given in Section III, an active primary cell
can support a constant data rate of $K_3$, where $K_3>0$
independent of $n$ and $m$.
\end{lemma}

The proof can be found in Appendix II.

\begin{lemma}\label{lemma6}
With the protocols given in Section III, the secondary tier can
deliver the primary packets to the intended primary destination
node in a sink cell at a constant data rate of $K_4$, where
$K_4>0$ independent of $n$ and $m$.
\end{lemma}

The proof can be found in Appendix II.

Based on \textit{Lemmas 1-2} and \textit{Lemmas 5-6}, we have the following theorem.

\begin{theorem}
With the protocols given in Section III, the primary tier can
achieve the following throughput per S-D pair and sum throughput
w.h.p.:
\begin{equation}\label{pthroughput1m}
    \lambda_p(n)=\Theta\left(\frac{1}{na_p(n)}\right)
\end{equation}
and
\begin{equation}\label{pthroughput2m}
    T_p(n)=\Theta\left(\frac{1}{a_p(n)}\right),
\end{equation}
when $a_p(n)\geq \sqrt 2\beta\log n/n$ and $a_p(n)=o(1)$.
\end{theorem}
\begin{proof}
From~\textit{Lemmas~\ref{lemma5}} and \textit{\ref{lemma6}}, we
know that a primary TX can pour its packets into the secondary
tier at rate $K=\min(K_3,K_4)$ w.h.p.. Similar to the proof of \textit{Theorem 1}, there are $\Theta(na_p(n))$ of primary source nodes, which take turns to transmit packets in each active primary cell w.h.p..
Therefore, the upper-bound of the throughput per S-D pair is of
$\Theta\left(K/(na_p(n))\right)=\Theta\left(1/(na_p(n))\right)$
w.h.p.. Next, we show that with the proposed protocols, the above
upper-bound is achievable. In the proposed
protocols, from \textit{Lemma 2} we know that a randomly selected designated relay node
for the primary packet in each primary cell is a secondary node
w.h.p. from \textit{Lemma~\ref{lemma2}}. As such, when a primary
cell is active, the current primary time slot is just used for the
primary source nodes in the primary cell to transmit their own
packets w.h.p.. Therefore, the achievable throughput per S-D pair
is of
$\Theta\left(pK/(na_p(n))\right)=\Theta\left(1/(na_p(n))\right)$ and thus a achievable sum throughput of
$\Theta(1/a_p)$ for the primary tier w.h.p.. This completes the proof.
\end{proof}

By setting $a_p(n)=\sqrt 2\beta\log n/n$, the primary tier can achieve the
following throughput per S-D pair and sum throughput w.h.p.:
\begin{equation}\label{compare2}
    \lambda_p(n)=\Theta\left(\frac{1}{\log n}\right)
\end{equation}
and
\begin{equation}
    T_p(n)=\Theta\left(\frac{n}{\log n}\right).
\end{equation}

From (\ref{compare2}), we can draw the same conclusion as that in the scenario with static secondary nodes, i.e., the per-node throughput scaling law of the primary tier can be improved from $\Theta\left(1/(\sqrt {n\log n})\right)$ as in the stand-alone network to $\Theta\left(1/\log n)\right)$ with the help of the secondary tier.

\vspace{0.25cm}
\noindent\textbf{Delay Analysis}

Based on the proposed supportive protocols, we know that the delay for each primary packet has two
components: i) the hop delay, which is the transmission time for two
hops (from the primary source node to a secondary relay node
and from the secondary relay node to the primary destination node); ii)
the queueing delay, which is the time a packet spends in the
relay-queue at the secondary node until it is delivered to its
destination. The hop delay is two primary time slots, which can be
considered as a constant independent of $m$ and $n$. Next, we
quantify the primary-tier delay performance by focusing on the expected queueing delay at the relay based on the two
mobility models described in Section II.C.

\subsubsection{The i.i.d. Mobility Model}
We have the following theorem regarding the delay of the primary
tier.

\begin{theorem}\label{pdelaym1}
With the protocols given in Section III, the primary tier can
achieve the following delay w.h.p. when $\beta\geq 2$:
\begin{equation}\label{result1}
    D_p(n)=\Theta(1).
\end{equation}
\end{theorem}

\begin{proof}
According to the secondary protocol, within the secondary tier we have $\Theta(m)$ secondary nodes act
as relays for the primary tier, each of them with a separate queue for each of the
primary S-D pairs. Therefore, the queueing delay is the expected
delay at a given relay-queue. By symmetry, all such relay-queues
incur the same delay w.h.p.. For convenience, we fix
one primary S-D pair and consider the $\Theta(m)$ secondary nodes together
as a virtual relay node as shown in~Fig.~\ref{virtualnode} without identifying which secondary node is used as the relay. As such, we can calculate the expected
delay at a relay-queue by analyzing the expected delay at the
virtual relay node. Denote the selected primary source node,
the selected primary destination node, and the virtual relay node as S, D, and R,
respectively. To calculate the expected delay at node R, we first have to characterize the arrival and departure processes.
A packet arrives at R when a) the primary cell containing S is
active, and b) S transmits a packet. According to the primary
protocol in Section III, the primary cell containing S becomes
active every $64$ primary time slots. Therefore, we consider $64$
primary time slots as an observation period, and treat the arrival
process as a Bernoulli process with rate $p$ ($0<p<1$). Similarly, packet
departure occurs when a) D is in a sink cell, and b) at least one
of the relay nodes that have the desired packets for D is in the
sink cell containing D. Let $q$ detnote the probability that event
b) occurs, which can be expressed as
\begin{eqnarray}
% \nonumber to remove numbering (before each equation)
  q &=& 1-(1-a_p(n))^M, \\\nonumber
   &\sim& 1-e^{-Ma_p(n)}, \\\nonumber
   &\rightarrow& 1,~\textrm{as}~n\rightarrow~\infty,~\textrm{for}~\beta\geq 2,
\end{eqnarray}
where $f\sim g$ means that $f$ and $g$ have the same limit when
$n\rightarrow\infty$, $M=\Theta(ma_p(n))$ denotes the number of
the secondary nodes that have desired packets for D in the sink
cell containing D and belong to Class I (Class II) if D is in a
sink cell at even (odd) time slots. As such, the departure process
is an asymptotically deterministic process with departure rate
$q=1$. Let $W_1$ denote the delay of the queue at the virtual
relay node based on the i.i.d. model. Thus, the queue at the
virtual relay node is an asymptotically Bernoulli/deterministic
queue, with the expected queueing delay given
by~\cite{Daduna:Queueing}
\begin{equation}
  E\{W_1\} = 64\frac{1-p}{q-p}\rightarrow 64,~\textrm{as}~n\rightarrow \infty,
\end{equation}
where $E\{\cdot\}$ denotes the expectation and the factor $64$ is
the length of one observation period. Note that the queueing
length of this asymptotically Bernoulli/deterministic queue is at
most one primary packet length w.h.p..

Next we need to verify that the relay-queue at each of the $\Theta(m)$
secondary nodes is stable over time. Note that based on the proposed protocol every secondary node removes
the outdated packets that have the SNs lower than the desired one
for D when it jumps into the sink cell containing D. Since the
queueing length at R can be upper-bounded by one, by considering the effect of storing outdated packets, the length of
the relay-queue at each secondary node can be upper-bounded by
\begin{equation}\label{queueing length}
    L=n+1
\end{equation}
where $n$ can be considered as an upper-bound for the inter-visit
time of the primary cell containing D, since
$(1-a_p(n))^n\rightarrow 0$ as $n\rightarrow \infty$. Thus, the
relay-queues at all secondary nodes are stable over time for each given $n$, which completes the
proof.
\end{proof}
\begin{figure}
\centering
\includegraphics[width=2.5in]{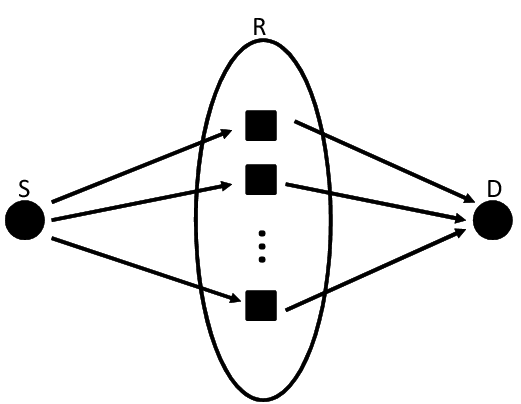}
\caption{Illustration of the virtual relay node R.}
\label{virtualnode}
\end{figure}

\subsubsection{The RW Mobility Model}

For the RW model, we have the following theorem regarding the
delay of the primary tier.
\begin{theorem}\label{pdelaym2}
With the protocols given in Section III, the primary tier can
achieve the following delay w.h.p. when $\beta\geq2$:
\begin{equation}\label{result2}
    D_p(n)=\Theta\left(\frac{1}{S}\right)=O\left(\frac{1}{a_p(n)}\right)
\end{equation}
where $S\geq a_p(n)$.
\end{theorem}

\begin{proof}
Like the proof in the i.i.d. mobility case, we fix a primary S-D
pair and consider the $\Theta(m)$ secondary nodes together as a virtual
relay node. Denote the selected primary source node, the selected primary destination
node, and the virtual relay node as S, D, and R, respectively.
Based on the proposed secondary protocol in Section III, each
secondary node maintains $Q=\tau$ queues for each primary S-D pair.
Equivalently, R also maintains $Q$ queues for each primary S-D pair where each queue is a concatenated one from $\Theta(m)$ small ones,
and the packet that arrives at time $t$ is stored in the $k$th queue,
where
$k=\left\{\lfloor\frac{t}{64}\rfloor~\textrm{mod}~\tau\right\}$.
By symmetry, all such queues incur the same expected delay.
Without loss of generality, we analyze the expected delay of the
$k$th queue by characterizing its arrival and departure processes.
A packet that arrives at time $t$ enters the $k$th queue when a)
the primary cell containing S is active, b) S transmits a packet,
and c) $\left\{\lfloor\frac{t}{64}\rfloor~\textrm{mod}
~\tau\right\}=k$. Consider $64\tau$ primary time slots as an
observation period. The arrival process is a Bernoulli process
with arrival rate $p$. Similarly, a packet departure occurs at
time $t$ when a) D is in a sink cell, b) at least one of the relay
nodes that have the desired packets for D is in the sink cell
containing D, and c)
$\left\{\lfloor\frac{t}{64}\rfloor~\textrm{mod} ~\tau\right\}=k$.
Let $q$ denote the probability that event b) occurs during one
observation period, which can be expressed as
\begin{eqnarray}
% \nonumber to remove numbering (before each equation)
  q &=& 1-\left(1-\prod_{i\in\mathcal{I}}{q_0p_{(x_i,y_i)(x_d,y_d)}(t_d)}\right), \\\nonumber
   &\geq& 1-(1-q_0(1-e^{-1})S)^M,\\\nonumber
   &\sim& 1-e^{-q_0(1-e^{-1})SM},\\\nonumber
   &\rightarrow& 1,~\textrm{as}~n\rightarrow~\infty,~\textrm{for}~\beta\geq 2,
\end{eqnarray}
where $\mathcal{I}$ denotes the set of the secondary nodes that
have the desired packets for D and belong to Class I (Class II) if D is
in a sink cell at even (odd) time slots; $(x_i,y_i)$ represents
the index of the RW-cell, in which the $i$th secondary node in
$\mathcal{I}$ is located when S sends the desired packet;
$(x_d,y_d)$ is the index of the RW-cell, in which D is located;
$t_d$ stands for the difference between the arrival time and the
departure time for the desired packet, which can be lower-bounded
by $64(\tau-1)$; and $q_0$ denotes the probability that a secondary
node is within the sink cell containing $D$ when it moves into
RW-cell $(x_d,y_d)$, which is given by $q_0=a_p(n)/S$. As such, the departure process is an asymptotically deterministic
process with departure rate $q=1$. Let $W_2$ denote the delay of
the queue at node R based on the RW model. Thus,
the queue at node R is an asymptotically
Bernoulli/deterministic queue, with the queueing delay given by
\begin{equation}
  E\{W_2\} = 64\tau\frac{1-p}{q-p}\sim64\tau=\Theta(\frac{1}{S}),
\end{equation}
where the factor $64\tau$ is the length of one observation period.
Since $S\geq a_p(n)$, we have $E\{W_2\}=O\left(1/a_p(n)\right)$.

Using the similar argument as in the i.i.d. case, we can
upper-bound the length of the $k$th relay-queue at any secondary
node by~(\ref{queueing length}) for any $k$. Thus, the
relay-queues at all secondary nodes are stable, which completes
the proof.
\end{proof}

In TABLE I, we compare our delay scaling results for the primary tier in the two-tier network setup with the optimal delay scaling for a stand-alone primary network (without the secondary tier), in which the throughput scalings for all scenarios are fixed to be $\Theta(1/\sqrt{n\log n})$. From TABLE I, we see that the primary tier achieves worse delay scaling in the presence of the static secondary nodes compared with the one without the secondary nodes. However, in the scenario with mobile secondary nodes, the delay scaling of the primary tier is significantly improved.

\begin{table*}
\caption{The delay scaling laws of the primary tier.}
\begin{center}
    \begin{tabular}{|c|c|c|c|c|}
    \hline
      \multirow {2}{*} & Stand-Alone  Network & \multicolumn{3}{|c|} {Supportive Two-Tier Network} \\ \hhline{|~|~|-|-|-|}
      &  \multirow{2}{*} {(Without Secondary Nodes)} &  \multirow{2}{*}{Static Secondary Nodes} & Mobile Secondary Nodes &  Mobile Secondary Nodes\\

      & & &  (i.i.d. Case) & (RW Case)\\ \hline
      Delay Scaling Law& $\Theta(\sqrt{n/\log n})$& $\Theta\left(\sqrt{n^{\beta-1}}\right)$& $\Theta(1)$& $O(\sqrt{n/\log n})$\\
      \hline
    \end{tabular}
\end{center}
\end{table*}

\vspace{0.25cm}
\noindent\textbf{Delay-Throughput Tradeoff}

For the RW model, we have the following delay-throughput tradeoff
for the primary tier by combining (\ref{pthroughput1}) and
(\ref{result2}).
\begin{equation}
    D_p(n)=O\left(n\lambda_p(n)\right),~~\textrm{for}~\lambda_p(n)=O\left(\frac{1}{\log n}\right).
\end{equation}

\begin{figure}
\centering
\includegraphics[width=2.5in]{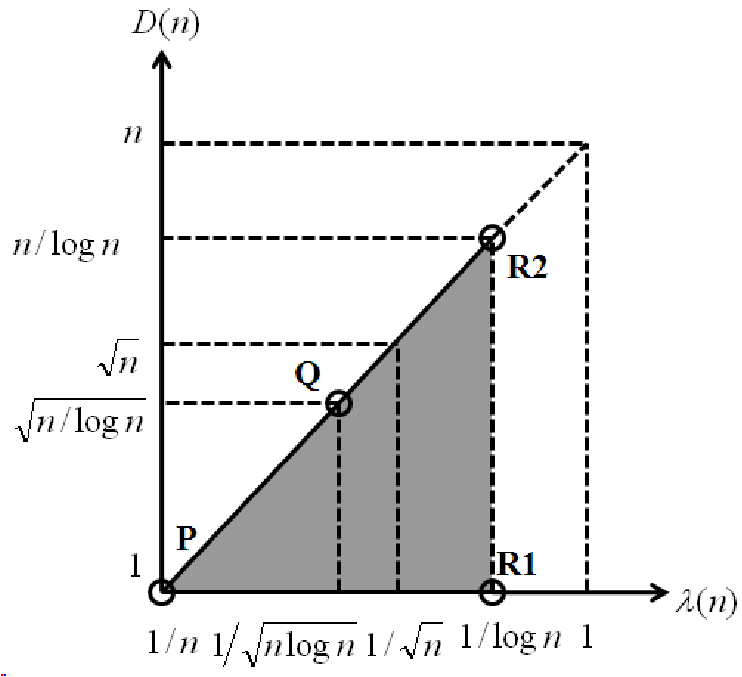}
\caption{Delay-throughput tradeoff for the primary tier with the aid of mobile secondary nodes.}
\label{tradeoff2}
\end{figure}

In Fig.~\ref{tradeoff2}, we draw the delay-throughput tradeoff for the primary tier with the aid of the mobile secondary nodes compared to the result for a static stand-alone network in~\cite{Gamal:Optimal}. In the figure, the line segments PR1 and PQ denote the delay-throughput tradeoffs for the primary tier with and without the secondary tier, respectively, where the segment PR1 is obtained based on the i.i.d. node mobility model. Any delay-throughput pair in the two segments can be achieved by adjusting the size of the primary cell. For the RW mobility model, any delay-throughput pair in the shade area can be achieved by adjusting the size of the primary cell and the RW step size. We see that with the aid of the mobile secondary nodes, both throughput and delay scaling laws of the primary tier can be improved in this scenario.

We see that for both mobility models, the delay-throughput tradeoffs for the primary tier
with the aid of the mobile secondary nodes are better than the optimal
delay-throughput tradeoff given in~\cite{Gamal:Optimal} for a
static stand-alone network. Particularly, the obtained delay-throughput tradeoff for the i.i.d. mobility model is essentially optimal for the supportive two-tier network setup, since the achievable constant delay scaling law is also the lower bound for any given throughput scaling on the order of $O(1/\log n)$. Note that the above throughput and delay analysis is based on the assumption $\beta\geq 2$, and we leave the case with $1<\beta<2$ in our future work.

\section{Throughput and Delay Analysis for the Secondary Tier}
\subsection{The Scenario with Static Secondary Nodes}

\vspace{0.25cm}
\noindent\textbf{Throughput Analysis}

In this section, we discuss the delay and throughput scaling laws
for the secondary tier. According to the protocol for the
secondary tier, we split the time frame into three equal-length
fractions and use one of them for the secondary packet
transmissions. Since the above time-sharing strategy only incurs a
constant penalty (i.e., 1/3) on the achievable throughput and
delay within the secondary tier, the throughput and delay scaling
laws are the same as those given in~\cite{Yin:Scaling}, which are
summarized by the following theorems.
\begin{theorem}
With the secondary protocol defined in Section III, the secondary
tier can achieve the following throughput per S-D pair and sum
throughput w.h.p.:
\begin{equation}\label{sthroughput1s}
    \lambda_s(m)=\Theta\left(\frac{1}{m\sqrt{a_s(m)}}\right)
\end{equation}
and
\begin{equation}\label{sthroughput2s}
    T_s(m)=\Theta\left(\frac{1}{\sqrt{a_s(m)}}\right),
\end{equation}
where $a_s(m)\geq 2\log m/m$ and the specific value of $a_s(m)$ is
determined by $a_p(n)$ as shown in Appendix III.
\end{theorem}
\vspace{0.25cm}
\noindent\textbf{Delay Analysis}

\begin{theorem}
With the secondary protocol defined in Section III, the packet
delay is given by
\begin{equation}\label{sdelay}
D_s(m)=\Theta\left(\frac{1}{\sqrt{a_s(m)}}\right).
\end{equation}
\end{theorem}
\vspace{0.25cm}
\noindent\textbf{Delay-Throughput Tradeoff}

Combining the results in~(\ref{sthroughput1s}) and (\ref{sdelay}),
the delay-throughput tradeoff for the secondary tier is given by
the following theorem.
\begin{theorem}
With the secondary protocol defined in Section III, the
delay-throughput tradeoff is
\begin{equation}
D_s(m)=\Theta(m\lambda_s(m)),~\textrm{for}~\lambda_s(m)=O\left(\frac{1}{\sqrt{m\log
m}}\right).
\end{equation}

\end{theorem}
For detailed proofs of the above theorems, please refer
to~\cite{Yin:Scaling}.

\subsection{The Scenario with Mobile Secondary Nodes}

When a secondary RX receives its own packets, it suffers from two interference terms from all active primary
TXs and all active secondary TXs. We can use a similar method as in the proof of \textit{Lemma 5} to prove that each of the two interference terms can be upper-bounded by a constant independent
of $m$ and $n$. Thus, the asymptotic results for a stand-alone network in~\cite{Michael:Capacity},~\cite{Gamal:Optimal} hold in this scenario. In the following, we summarize these results for completeness.

\vspace{0.25cm}
\noindent\textbf{Throughput Analysis}

We have the following theorem regarding the throughput scaling law for the secondary tier.
\begin{theorem}
With the protocols given in Section III, the secondary tier can
achieve the following throughput per S-D pair and sum throughput
w.h.p.:
\begin{equation}\label{sthroughput1m}
    \lambda_s(m)=\Theta(1)
\end{equation}
and
\begin{equation}\label{sthroughput2m}
    T_s(m)=\Theta(m).
\end{equation}
\end{theorem}
\vspace{0.25cm}
\noindent\textbf{Delay Analysis}

Next, we provide the delay scaling laws of the secondary tier for
the two mobility models as discussed in Section II.C.

\begin{theorem}
With the protocols given in Section III, the secondary tier can
achieve the following delay w.h.p. based on the i.i.d. mobility
model:
\begin{equation}\label{sdelay1m}
    D_s(m)=\Theta(m).
\end{equation}
\end{theorem}

\begin{theorem}
With the protocols given in Section III, the secondary tier can
achieve the following delay w.h.p. based on the RW model:
\begin{equation}\label{sdelay2m}
    D_s(m)=\Theta\left(m^2S\log\frac{1}{S}\right).
\end{equation}
\end{theorem}

Note that~(\ref{sdelay2m}) is a generalized result for $S\geq 1/m$.
When $S=1/m$, the delay $D_s(m)=\Theta(m\log m)$ is the same as
that in~\cite{Gamal:Optimal}.

\section{Conclusion}

In this paper, we studied the throughput and delay scaling laws
for a supportive two-tier network, where the secondary
tier is willing to relay packets for the primary tier.
When the secondary tier has a much higher density, the primary
tier can achieve a better throughput scaling law compared to non-interactive overlaid
networks. The delay scaling law for the primary tier can also be improved when then the secondary nodes are mobile. Meanwhile, the secondary tier can still achieve the same delay and throughput tradeoff as in a stand-alone network.

\begin{center}
    {\bf APPENDIX I}\\
    {\bf Proofs of \textit{Lemmas 3} and \textit{4}}
\end{center}

\begin{proof}[Proof of \textit{Lemma 3}]
Assume that at a given moment, there are $K_p$ active primary
cells. The rate of the $i$th active primary cell is given by
\begin{equation}\label{shannon rate}
    R_p(i)=\frac{1}{64}\log\left(1+\frac{P_p(i)g\left(||X_{p,tx}-X_{p,rx}||\right)}{N_0+I_p(i)+I_{sp}(i)}\right)
\end{equation}
where $P_p(i)=Pa_{p}^{\frac{\alpha}{2}}(n)$ and $\frac{1}{64}$ denotes the rate loss due to the 64-TDMA
transmission of primary cells. In the surrounding of the $i$th
primary cell, there are 8 primary interferers with a distance of
at least $6\sqrt{a_p}$ and 16 primary interferers with a distance
of at least $13\sqrt{a_p}$, and so on. As such, the $I_p(i)$ is upper-bounded by
\begin{eqnarray}\label{primary interf}
% \nonumber to remove numbering (before each equation)
  I_p(i) &=& \sum_{k=1,k\neq i}^{K_p}P_p g(\left||X_{p,tx}(k)-X_{p,rx}(i)||\right) \\\nonumber
   &<& P\sum_{t=1}^{\infty}8t(7t-1)^{-\alpha}\triangleq A.
\end{eqnarray}
Next, we discuss the interference $I_{sp}(i)$ from secondary
transmitting interferers to the $i$th primary RX. We consider the
following two case:
\begin{description}
  \item[Case I]: The secondary tier transmits either the secondary packets to the next hop or the primacy packets to the next secondary relay nodes, i.e., in the first or secondary subframes.
  \item[Case II]: The secondary tier delivers the data packets to the primary destination nodes, i.e., in the third secondary subframe.
\end{description}
In Case I, assume that there are $K_s$ active secondary cells, which means that the number of the active secondary TXs is also $K_s$. Since a minimum distance $\sqrt{a_s}$ can be guaranteed
from all secondary transmitting interferers to the primary RXs in
the preservation regions, $I_{sp}(i)$ is upper-bounded by
\begin{eqnarray}\label{secondary interf1}
% \nonumber to remove numbering (before each equation)
  I_{sp}^{I}(i) &=& \sum_{k=1,k\neq i}^{K_s}P_s g(\left||X_{s,tx}(k)-X_{p,rx}(i)||\right)\\\nonumber
   &<& P\sum_{t=1}^{\infty}8t(7t-6)^{-\alpha}\triangleq B.
\end{eqnarray}

In Case II, there are $K_p$ collection regions and thus $K_p$ active secondary TXs. In the surrounding of the $i$th
primary cell, there are 2 secondary interferers with a distance of
at least $2\sqrt{a_p}$ and 4 secondary interferers with a distance
of at least $9\sqrt{a_p}$, and so on. Then, $I_{sp}(i)$ is upper-bounded by
\begin{eqnarray}\label{secondary interf2}
% \nonumber to remove numbering (before each equation)
  I_{sp}^{II}(i) &=& \sum_{k=1,k\neq i}^{K_p}P_p g(\left||X_{s,tx}(k)-X_{p,rx}(i)||\right) \\\nonumber
  &<& P\sum_{t=1}^{\infty}2t(7t-5)^{-\alpha}\triangleq C.
\end{eqnarray}
Given $B>A$ and $B>C$, we have
\begin{equation}\label{bounded rate}
    R_p(i)>\frac{1}{64}\log\left(1+\frac{P(\sqrt 5)^{-\alpha}}{N_0+2P\sum_{t=1}^\infty 8t(7t-6)^{-\alpha}}\right).
\end{equation}
Since $\sum_{t=1}^\infty 8t(7t-6)^{-\alpha}$ converges to a
constant for $\alpha>2$, there exists a constant $K_1>0$ such that
$R_p(i)>K_1$. This completes the proof.
\end{proof}

Note that from (\ref{shannon rate}), the rate of the i-th primary cell actually decreases due to the extra interference from the secondary tier compared with the case without the secondary tier. However, the scaling law of the rate does not change (which is still a constant scaling).

\begin{proof}[Proof of \textit{Lemma 4}]
The proof is similar to that for \textit{Lemma 3}. When a primary RX
receives packets from its surrounding secondary nodes, it suffers
from two interference terms from all active primary TXs and all active
secondary TXs, either of which can be upper-bounded by a constant
independent of $n$ and $m$. Thus there is a constant rate $K_2$, at which
the secondary tier can deliver packets to the intended primary destination node.
\end{proof}

\begin{center}
    {\bf APPENDIX II}\\
    {\bf Proofs of \textit{Lemmas 5} and \textit{6}}
\end{center}

\begin{proof}[Proof of \textit{Lemma 5}]
Assume that at a given moment, there are $K_p$ active primary
cells. The supported rate of the $i$th active primary cell is given by
\begin{equation}
    R_p(i)=\frac{1}{64}\log\left(1+\frac{P_p(i)g\left(||X_{p,tx}-X_{p,rx}||\right)}{N_0+I_p(i)+I_{sp}(i)}\right)
\end{equation}
where $\frac{1}{64}$ denotes the rate loss due to the 64-TDMA
transmission of primary cells. In the surrounding of the $i$th
primary cell, there are 8 primary interferers with a distance of
at least $6\sqrt{a_p}$ and 16 primary interferers with a distance
of at least $13\sqrt{a_p}$, and so on. As such, the $I_p(i)$ is
upper-bounded by
\begin{eqnarray}
% \nonumber to remove numbering (before each equation)
  I_p(i) &=& \sum_{k=1,k\neq i}^{K_p}P_p g(\left||X_{p,tx}(k)-X_{p,rx}(i)||\right) \\\nonumber
   &<& P\sum_{t=1}^{\infty}8t(7t-1)^{-\alpha}\triangleq A.
\end{eqnarray}
Next, we discuss the interference $I_{sp}(i)$ from secondary
transmitting interferers to the $i$th primary RX. According to the
proposed secondary protocol, the secondary nodes are divided into
two classes: Class I and Class II, which operate over the switched
timing relationships with the odd and the even time slots. Without
the loss of generality, we consider the interference $I_{sp}(i)$
from secondary transmitting interferers to the $i$th primary RX at
the odd time slots. Assume that there are $K_s$ active secondary
cells, which means that the number of the active secondary TXs of Class
I is $K_s$. Since a minimum distance $\sqrt{a_s}$ can be
guaranteed from all secondary transmitting interferers of Class I to the
primary RXs in the preservation regions, the interference from the
active secondary TXs of Class I, $I_{sp}^I(i)$, is upper-bounded by
\begin{eqnarray}\label{secondary interf11}
% \nonumber to remove numbering (before each equation)
  I_{sp}^{I}(i) &=& \sum_{k=1,k\neq i}^{K_s}P_s g(\left||X_{s,tx}(k)-X_{p,rx}(i)||\right)\\\nonumber
   &<& P\sum_{t=1}^{\infty}8t(7t-6)^{-\alpha}\triangleq B.
\end{eqnarray}
Furthermore, there are $K_p$ collection regions, which means that the number
of the active secondary TXs of Class II is $K_p$. Since
a minimum distance $2\sqrt{a_p}$ can be guaranteed from all
secondary transmitting interferes of Class II to the primary RXs in the
preservation regions, the interference from the active secondary
TXs of Class II, $I_{sp}^{II}(i)$, is upper-bounded by
\begin{eqnarray}
% \nonumber to remove numbering (before each equation)
  I_{sp}^{II}(i) &=& \sum_{k=1,k\neq i}^{K_p}P_p g(\left||X_{p,tx}(k)-X_{p,rx}(i)||\right) \\\nonumber
   &<& P\sum_{t=1}^{\infty}8t(7t-5)^{-\alpha}\triangleq C.
\end{eqnarray}

Given $B>A$ and $B>C$,
we have
\begin{eqnarray}
% \nonumber to remove numbering (before each equation)
  R_p(i) &=& \frac{1}{64}\log\left(1+\frac{P_p(i)g\left(||X_{p,tx}-X_{p,rx}||\right)}{N_0+I_p(i)+I_{sp}^{I}(i)+I_{sp}^{II}(i)}\right) \\\nonumber
   &>& \frac{1}{64}\log\left(1+\frac{P(\sqrt 5)^{-\alpha}}{N_0+3P\sum_{t=1}^\infty 8t(7t-6)^{-\alpha}}\right).
\end{eqnarray}
Since $\sum_{t=1}^\infty 8t(7t-6)^{-\alpha}$ converges to a
constant for $\alpha>2$, there exists a constant $K_3>0$ such that
$R_p(i)>K_3$. This completes the proof.
\end{proof}

\begin{proof}[Proof of \textit{Lemma 6}]
The proof is similar to that for \textit{Lemma 5}. When a primary RX
receives packets from its surrounding secondary nodes, it suffers
from three interference terms from all active primary TXs, all active
secondary TXs of Class I, and all active secondary TXs of Class II, each of which can be upper-bounded by a constant
independent of $n$ and $m$. Thus, there is a constant rate $K_4$, at which
the secondary tier can deliver packets to the intended primary destination node.
\end{proof}

\begin{center}
    {\bf APPENDIX III}\\
    {\bf Derivation of (\ref{cellsize})}
\end{center}

We know that given $a_p(n)\geq \sqrt 2\beta\log n/n$, the maximum throughput
per S-D pair for the primary tier is
$\Theta\left(\frac{1}{na_p(n)}\right)$. Since a primary packet is
divided into $N$ segments and then routed by $N$ parallel S-D
paths within the secondary tier, the supported rate for each
secondary S-D pair is required to be
$\Theta\left(\frac{1}{Nna_p(n)}\right)=\Theta\left(\frac{\sqrt{\log
m}}{\sqrt mna_p(n)}\right)$. As such, based
on~(\ref{sthroughput1s}), the corresponding secondary cell size
$a_s(m)$ needs to be set as
\begin{equation*}
    a_s(m)=\frac{n^2a_p^2(n)}{m\log m}
\end{equation*}
where we have $a_s(m)\geq 2\log m/m$ when $a_p(n)\geq\sqrt 2\beta\log n/n$.


\begin{thebibliography}{19}

\bibitem{Gupta:Capacity}P. Gupta and P. R. Kumar, ``The capacity of wireless networks,'' \textit{IEEE Transactions on Information Theory}, vol. 46, pp. 388-404, Mar. 2000.
\bibitem{Francheschetti:Closing}M. Francheschetti, O. Dousse, D. Tse, and P. Thiran, ``Closing the gap in the capacity of random wireless networks via percolation theory,'' \textit{IEEE Transactions on Information Theory}, vol. 53, no. 3., pp. 1009-1018, Mar. 2007.
\bibitem{Josan:Throughput}A. Josan, M. Liu, D. L. Neuhoff, and S. S. Pradhan, ``Throughput scaling in random wireless networks: a non-hierarchical multipath routing strategy,'' Preprint. Oct. 2007. [Online]. Available: http://arxiv.org/pdf/0710.1626.
\bibitem{Ozgur:old}A. Ozgur, O. Leveque, and D. Tse, ``Hierarchical cooperation achieves optimal capacity scaling in ad hoc networks,'' \textit{IEEE Transactions on Information Theory}, vol 53, no. 10, pp. 3549 - 3572, Oct. 2007.
\bibitem{Matthias:Mobility}M. Grossglauser and D. N. C. Tse, ``Mobility increases the capacity of ad hoc wireless network,'' \textit{IEEE/ACM Transaction on Networking}, vol. 10, no. 4, pp. 477-486, Aug. 2002.
\bibitem{Michael:Capacity}M. J. Neely and E. Modiano, ``Capacity and delay tradeoffs for ad-hoc mobile networks,'' \textit{IEEE Transactions on Information Theory}, vol. 51, no. 6, pp. 1917-1936, June 2005.
\bibitem{Gamal:Optimal}A. E. Gamal, J. Mammen, B. Prabhakar, and D. Shah, ``Optimal throughput-delay scaling in wirless networks--part I: the fluid model,'' \textit{IEEE Transaction on Information Theory}, vol. 52, no. 6, pp. 2568-2592, June 2006.
\bibitem{Ying: Optimal}L. Ying, S. Yang, and R. Srikant, ``Optimal delay-throughput tradeoffs in mobile ad hoc networks,'' \textit{IEEE Transaction on Information Theory}, vol. 54, no. 9, Sept. 2008.
\bibitem{Nikhil:Capacity}N. Bansal and Z. Liu, ``Capacity, delay and mobility in wireless ad-hoc networks,'' \textit{IEEE INFOCOM 2003}, vol. 2, pp. 1553-1563, Mar. 2003.
\bibitem{Ozgur:Throughput}A. Ozgur and O. Leveque, ``Throughput-delay tradeoff for hierarchical cooperation in ad hoc wireless networks,'' \textit{IEEE Transactions on Information Theory}, vol. 56, no. 3, pp. 1369-1377, Mar. 2010.
\bibitem{Jeon:Cognitive}S. Jeon, N. Devroye, M. Vu, S. Chung, and V. Tarokh, ``Cognitive networks achieve throughput scaling of a homogeneous network,'' preprint. Jan. 2008. [Online]. Available: http://arxiv.org/pdf/0801:0938.
\bibitem{Yin:Scaling}C. Yin, L. Gao, and S. Cui, ``Scaling laws for overlaid wireless networks: a cognitive radio network vs. a primary network,'' \textit{IEEE/ACM Transaction on Networking}, vol. 18, no. 4, pp. 1317-1329, Aug. 2010.
\bibitem{Goldsmith:Breaking}A. Goldsmith, S. Jafar, I. Maric, and S. Srinivasa, ``Breaking spectrum gridlock with cognitive radios: an information theoretic perspective,'' \textit{IEEE Proceedings}, vol. 97, no. 5, pp. 894-914, May 2009.
\bibitem{Devroye:Limits}N. Devroye, P. Mitran, and V. Tarokh, ``Limits on communications in a cognitive radio channel,'' \textit{IEEE Communication Magzine}, vol. 44, no. 6, pp. 44-49, June 2006.
\bibitem{Devroye:Achievable}N. Devroye, P. Mitran, and V. Tarokh, ``Achievable rates in cognitive radio channels,'' \textit{IEEE Transactions on Information Theory}, vol. 52, no. 5, pp. 1813-1827, Apr. 2006.
\bibitem{Liu:Pernode}D. Xu and X. Liu, ``Per User Throughput in Large Wireless Networks,'' \textit{IEEE Communications Society Conference on Sensor and Ad Hoc Communications and Networks (SECON)}, pp. 91-99, June 2008.
\bibitem{Daduna:Queueing}H. Daduna, \textit{Queueing Networks with Discrete Time Scale}, Springer, 2001.
\bibitem{Aldous:Reversible}D. Aldous and J. Fill, ``Reversible Markov chain and random walks on graph,'' [online]. Available: http://www.stat.berkeley.edu/users/aldous/RWG/book.html.
\end{thebibliography}
\end{document}